\documentclass[a4paper,UKenglish]{lipics}
\usepackage{microtype}%if unwanted, comment out or use option "draft"

\newcommand{\set}[1]{\left\{ #1 \right\}}

\newcommand{\C}{\mathbb{C}}

\newcommand{\R}{\mathbb{R}}
\newcommand{\Q}{\mathbb{Q}}
\newcommand{\A}{\mathbb{A}}
\newcommand{\Z}{\mathbb{Z}}
\newcommand{\N}{\mathbb{N}}

\renewcommand{\Dot}[1]{\langle #1 \rangle}
\renewcommand{\Re}[1]{\text{Re}\left( #1 \right)}

\renewcommand{\O}{\mathcal{O}}
\renewcommand{\P}{\mathcal{P}}

\usepackage{tikz}
\usetikzlibrary{backgrounds, patterns}

\title{Semialgebraic Invariant Synthesis for the Kannan-Lipton Orbit Problem}
\author[1]{Nathana{\"e}l Fijalkow}
\author[2]{Pierre Ohlmann}
\author[1,3]{Jo{\"e}l Ouaknine}
\author[1]{Amaury Pouly}
\author[1]{James Worrell}
\affil[1]{Department of Computer Science, Oxford University, UK}
\affil[2]{\'{E}cole Normale Sup{\'e}rieure de Lyon, France}
\affil[3]{Max Planck Institute for Software Systems (MPI-SWS), Saarland Informatics Campus, Germany}

\authorrunning{N. Fijalkow, P. Ohlmann, J. Ouaknine, A. Pouly, and J. Worrell} 

\Copyright{Nathana{\"e}l Fijalkow, Pierre Ohlmann, Jo{\"e}l Ouaknine, Amaury Pouly, and James Worrell}

% mandatory: Please provide 1-5 keywords
\serieslogo{}%please provide filename (without suffix)
\volumeinfo%(easychair interface)
  {\ }% editors
  {2}% number of editors: 1, 2, ....
  {Conference title on which this volume is based on}% event
  {1}% volume
  {1}% issue
  {1}% starting page number
\EventShortName{}
\DOI{10.4230/LIPIcs.xxx.yyy.p}% to be completed by the volume editor
\begin{document}

\maketitle

\begin{abstract}

The \emph{Orbit Problem} consists of determining, given a linear
transformation $A$ on $\mathbb{Q}^d$, together with vectors $x$ and
$y$, whether the orbit of $x$ under repeated applications of $A$ can
ever reach $y$. This problem was famously shown to be decidable by
Kannan and Lipton in the 1980s.

In this paper, we are concerned with the problem of synthesising
suitable \emph{invariants} $\mathcal{P} \subseteq \mathbb{R}^d$,
\emph{i.e.}, sets that are stable under $A$ and contain $x$ and not $y$, 
thereby providing compact and versatile certificates of non-reachability. We show that
whether a given instance of the Orbit Problem admits a semialgebraic
invariant is decidable, and moreover in positive instances we provide
an algorithm to synthesise suitable invariants of polynomial size.

It is worth noting that the existence of \emph{semilinear} invariants,
on the other hand, is (to the best of our knowledge) not known to be decidable.
\end{abstract}

\section{Introduction}
The \emph{Orbit Problem} was introduced by Kannan and Lipton 
in the seminal papers~\cite{KL80,KL86}, and shown there to be decidable in polynomial time,
answering in the process a decade-old open problem of Harrison on
accessibility for linear sequential machines~\cite{H69}.
The Orbit Problem can be stated as follows:

\begin{quote}
Given a square matrix
$A \in \mathbb{Q}^{d \times d}$ together with vectors $x, y \in \mathbb{Q}^d$,
decide whether there exists a non-negative integer $n$ such that $A^n
x = y$. 
\end{quote}

In other words, if one considers the discrete `orbit' of the vector
$x$ under repeated applications of the linear transformation $A$, does
the orbit ever hit the target $y$? Although it is not \emph{a priori}
obvious that this problem is even decidable, Kannan and Lipton showed
that it can in fact be solved in polynomial time, by making use of
spectral techniques as well as some sophisticated results from
algebraic number theory.

In instances of non-reachability, a natural and interesting question is
whether one can produce a suitable \emph{invariant} as certificate,
\emph{i.e.}, a set $\mathcal{P} \subseteq \mathbb{R}^d$ that is stable under
$A$ (in the sense that $A \mathcal{P} \subseteq \mathcal{P}$) and such that
$x \in \mathcal{P}$ and $y \notin \mathcal{P}$. The existence of such
an invariant then immediately entails by induction that the orbit of
$x$ does indeed avoid $y$.

Invariants appear in a wide range of contexts, from gauge
theory, dynamical systems, and control theory in physics, mathematics,
and engineering to program verification, static analysis, abstract
interpretation, and programming language semantics (among others) in
computer science. Automated invariant synthesis is a topic of active
current research, particularly in the fields of theorem proving and
program verification; in the latter, for example, one might imagine
that $y$ corresponds to a faulty or undesirable program state, and an
invariant $\mathcal{P}$ as described above amounts to a succinct
`safety' certificate (here the program or procedure in question
corresponds to a simple \textsc{while} loop with linear updates).

The widespread use of invariants should not come as a surprise. In
addition to their obvious advantage in constituting easily
understandable safety certificates, their inductive nature
makes them ideally suited to modular reasoning, often allowing one to analyse
complex systems by breaking them down into simpler parts, each of
which can then be handled in isolation. Invariants, viewed as safety
certificates, also enable one to reason over large sets of program states
rather than individual instances: in the context of the Orbit Problem,
for example, an invariant $\mathcal{P} \subseteq \mathbb{R}^d$ such
that $x \in \mathcal{P}$ and $y \notin \mathcal{P}$ doesn't merely
certify that $y$ is not reachable from $x$, but in fact guarantees
that from \emph{any} starting point $x' \in \mathcal{P}$, it is
impossible to reach \emph{any} of the points $y' \in \mathbb{R}^d \setminus
\mathcal{P}$.

In general, when searching for invariants, one almost always fixes
ahead of time a class of suitable potential candidates. Indeed, absent
such a restriction, one would point out that the orbit $\mathcal{O}(x)
= \{ A^n x : n \geq 0 \}$ is always by definition stable under $A$,
and in instances of non-reachability will therefore always constitute
a safety invariant. Such an invariant will however often not be of
much use, as it will usually lack good algorithmic properties;
for example, as observed in~\cite{KL86}, in dimension $d=5$ and
higher, the question of whether the orbit $\mathcal{O}(x)$ reaches a
given $(d-1)$-dimensional hyperplane corresponds precisely to the
famous \emph{Skolem Problem} (of whether an order-$d$ linear
recurrence sequence over the integers has a zero), whose decidability
has been open for over 80 years~\cite{Tao08}.

Thus let us assume that we are given a domain $\mathbf{D} \subseteq
2^{\mathbb{R}^d}$ of suitable potential invariants. At a minimum, one
would require that the relevant stability and safety conditions (\emph{i.e.},
for any $\mathcal{P} \in \mathbf{D}$, whether
$A \mathcal{P} \subseteq \mathcal{P}$, $x \in \mathcal{P}$, and
$y \notin \mathcal{P}$) be algorithmically checkable (with reasonable
complexity). The following natural questions then arise:

\begin{enumerate}
\item In instances of non-reachability, does a suitable invariant in $\mathbf{D}$
\emph{always} exist?

\item If not, can we characterise the exceptional instances in some way?

\item In instances of non-reachability, can we algorithmically determine whether
a suitable invariant in $\mathbf{D}$ exists, and when this is the case
can we moreover synthesise such an invariant?

\end{enumerate}

\noindent (\textbf{1}) and (\textbf{3}) are usually referred to as \emph{completeness}
and \emph{relative completeness} respectively, whereas (\textbf{2})~attempts to
measure the extent to which completeness fails.

\medskip\noindent\textbf{Main results.}
The main results of this paper concern the synthesis of semialgebraic
invariants for non-reachability instances of the Kannan-Lipton Orbit
Problem, where the input is provided as a triple $(A,x,y)$ with all entries rational,
and can be summarised as follows:

\begin{itemize}

\item We prove that whether a
suitable semialgebraic\footnote{A semialgebraic set is the set of
solutions of a Boolean combination of polynomial inequalities, with
the polynomials in question having integer coefficients.} invariant
exists or not is decidable in polynomial space, and moreover in
positive instances we show how to synthesise a suitable invariant of
polynomial size in polynomial space.

\item We provide a simple characterisation of instances of non-reachability for which
there does not exist a suitable semialgebraic invariant, and show that
such instances are very `rare', in a measure-theoretic sense.

\end{itemize}

Since the existence of suitable semialgebraic invariants for the Orbit
Problem does not coincide precisely with non-reachability, our proof
necessarily departs substantially from that given by Kannan and Lipton
in~\cite{KL80,KL86}. In particular, handling negative instances relies
upon certain topological and geometrical insights into the structure
of semialgebraic sets, and positive instances require the explicit
construction of suitable semialgebraic invariants of polynomial
size. We achieve this by making use of techniques from algebraic
number theory such as Kronecker's Theorem on inhomogeneous
simultaneous Diophantine approximation, and Masser's deep results on
multiplicative relations among algebraic numbers.

The following three examples
illustrate a range of phenomena that arise
in searching for semialgebraic invariants.
\begin{example}
\label{ex:one}
Consider the matrix
\[ A = {\textstyle\frac{1}{5}}\begin{pmatrix}
       4 & - 3\\[2pt] 3
       & 4 \end{pmatrix} \, . \] Matrix $A$ defines a
       counterclockwise rotation around the origin by angle
       $\arctan(3/5)$, which is an irrational multiple of $\pi$.  Thus
       the topological closure of the orbit
       $\mathcal{O}=\{x,Ax,A^2x,\ldots\}$ is a circle in
       $\mathbb{R}^2$.  If $y\not\in \overline{\mathcal{O}}$ then
       $\overline{\mathcal{O}}$ itself is clearly a suitable semialgebraic invariant.
       On other hand, it can be shown that if
       $y\in \overline{\mathcal{O}} \setminus \mathcal{O}$ then there does not exist a
       suitable semialgebraic invariant. (In passing, it is also not difficult
       to see that the only polygons $\mathcal{P}$ that are
       invariant under $A$ are $\emptyset$, $\{(0,0)\}$, and
       $\mathbb{R}^2$.)  More general orthogonal matrices can be
       handled along similar lines to the present case, but the
       analysis is substantially more involved.  In general, the only
       cases in which $y\not\in \mathcal{O}$ but there need not be a
       semialgebraic invariant are when the matrix $A$ is
       diagonalisable and all eigenvalues have modulus one, as in the
       case at hand.
\end{example}

\begin{example}
\label{ex:two}
Consider the matrix
\[ A = \frac{4}{25} \begin{pmatrix}
       4 & -3 & 4 & -3\\[2pt]
        3 & 4 & 3 & 4  \\[2pt]
               0 & 0 & 4 & -3 \\[2pt]
0 & 0 &3 & 4
\end{pmatrix} \] Matrix $A$ has spectral radius $\frac{4}{5}$ and so
$A^nx$ converges to $0$ for any initial vector $x \in \mathbb{Q}^4$.
Given a non-zero target $y \in \mathbb{Q}^4$ that does not lie
in the orbit $x,Ax,A^2x,\ldots$, a natural candidate for an invariant
is an initial segment of the orbit, together with some neighbourhood
$\mathcal{N}$ of the origin in $\mathbb{R}^4$ that excludes $y$ and is
invariant under $A$.  Note though that $A$ is not contractive with
respect to either the $1$-norm or the $2$-norm, so we cannot simply
take $\mathcal{N}$ to be a ball of suitably small radius with
respect to either of these norms.  However, for $\varepsilon>0$, the
set
\[ \mathcal{N}_{\varepsilon} = \left\{u \in \mathbb{R}^4 : u_1^2 + u_2^2 \leq
  \varepsilon^2 \wedge u_3^2+u_4^2 \leq
  \textstyle\frac{1}{16}\varepsilon^2 \right\} \]
\emph{is} invariant under $A$.  Thus we obtain a semialgebraic invariant
as the union of $\mathcal{N}_{\varepsilon}$, where $\varepsilon$ is
chosen sufficiently small such that $y\not\in \mathcal{N}_{\varepsilon}$,
together with an (easily computable) initial segment of the orbit
$x,Ax,A^2x,\ldots$ comprising all points in the orbit that lie outside
$\mathcal{N}_{\varepsilon}$.
\end{example}

\begin{example}
\label{ex:three}
Consider the following scaled version of the matrix from the previous example:
\[ A=
\frac{1}{5} \begin{pmatrix}
       4 & -3 & 4 & -3\\[2pt]
        3 & 4 & 3 & 4  \\[2pt]
               0 & 0 & 4 & -3 \\[2pt]
0 & 0 & 3 & 4
\end{pmatrix} \, .\]  Note that $A$ is a
non-diagonalisable matrix with spectral radius $1$.  
Example~\ref{ex:one} concerned an orthogonal matrix, while
the matrix in
Example~\ref{ex:two} was (morally speaking, if not literally)
length-decreasing.  Here, by contrast, the idea is to identify a subset
$\mathcal{Q}\subseteq \mathbb{R}^4$ that is invariant under $A$,
together with a ``length measure'' $f:\mathcal{Q}\rightarrow\mathbb{R}$
that increases under application of $A$.  Fixing a constant $c>0$, such
a set is
\[ \mathcal{Q} = \left\{
u \in \mathbb{R}^4 : u_1^2+u_2^2 \geq c \wedge u_1u_3+u_2u_4
\geq 0 \right\} \]
with length measure $f(u)=u_1^2+u_2^2$.  A key property of
$\mathcal{Q}$ is that for any vector $x\in \mathbb{R}^4$ such that
$x_3\neq 0$ or $x_4\neq 0$, the orbit $x,Ax,A^2x,\ldots$ eventually
enters $\mathcal{Q}$.  By choosing $c$ suitably large, we can exclude
$y$ from $\mathcal{Q}$.  Thus we obtain an invariant as the union of
$\mathcal{Q}$ and an appropriate finite intitial segment of the orbit
$x,Ax,A^2x,\ldots$.
\end{example}

We would like to draw the reader's attention to the critical role
played by the underlying domain $\mathbf{D}$ of potential
invariants. In the examples above as well as the rest of this paper,
we focus exclusively on the domain of semialgebraic sets. However one
might naturally consider instead the domain of \emph{semilinear} sets,
\emph{i.e.}, sets defined by Boolean combinations of linear inequalities with
integer coefficients, or equivalently consisting of finite unions of
(bounded or unbounded) rational polytopes. As pointed out above, in
Example~1 no non-trivial instance admits a semilinear invariant,
whereas one can show that in Example~2 semilinear invariants can
always be found. Interestingly, the question of relative completeness 
(\emph{i.e.}, determining in general whether or not a suitable semilinear
invariant exists in non-reachability instances) is not known to be
decidable, and appears to be a challenging problem.

\section{Preliminaries}
It is convenient in this paper to work over the field of (complex) algebraic numbers,
denoted $\A$. All standard algebraic operations, such as sums,
products, root-finding of polynomials and computing Jordan normal
forms of matrices with algebraic entries can be performed effectively;
we refer the reader to~\cite{COW16} for more details on the matter.

An \emph{instance of the Orbit Problem}, or \emph{Orbit instance} for
short, is given by a square matrix $A \in \A^{d \times d}$ and two
vectors $x, y \in \A^d$. The triple $(A,x,y)$ is a 
\emph{reachability} instance if there is
$n \in \N$ such that $A^n x = y$, and otherwise is
a \emph{non-reachability} instance.

We are interested in non-reachability certificates given as invariants.
Formally, given an Orbit instance $(A,x,y)$ in dimension $d$, a
set $\mathcal{P} \subseteq \mathbb{C}^d$ is a \emph{non-reachability invariant} if
$A \mathcal{P} \subseteq \mathcal{P}$, $x \in \mathcal{P}$, and $y \notin \mathcal{P}$.

For the remainder of this paper, we focus
on \emph{semialgebraic} invariants. Identifying $\mathbb{C}^d$ with
$\mathbb{R}^{2d}$, a set $\mathcal{P}$ is semialgebraic if it is the
set of real solutions of some Boolean combination of polynomial
inequalities with integer coefficients.

%We extend the notion of semialgebraic sets to subsets of
%$\mathbb{C}^n$ using the natural identification of $\mathbb{C}$ with
%$\mathbb{R}^2$:  we say that $S\subseteq \mathbb{C}^n$ is
%semialgebraic if there exists a semialgebraic set
%$R\subseteq \mathbb{R}^{2n}$ such that
%\[ S = \{ z \in \mathbb{C}^n : (\mathrm{Re}(z_1),\mathrm{Im}(z_1),\ldots,
%\mathrm{Re}(z_n),\mathrm{Im}(z_n)) \in R \} \, . \]
 
A central result about semialgebraic sets is the Tarski-Seidenberg
Theorem: if $S\subseteq\mathbb{R}^{n+1}$ is semialgebraic then the
image $\pi(S)$ under the projection
$\pi:\mathbb{R}^{n+1}\rightarrow\mathbb{R}^n$, where
$\pi(x_1,\ldots,x_{n+1})=(x_1,\ldots,x_n)$, is also semialgebraic.
Among the consequences of this result is the fact that the topological
closure of a semialgebraic set (in either $\mathbb{R}^n$ or
$\mathbb{C}^n$) is again semialgebraic.

\section{Semialgebraic Invariants}
Our main result is the following.

\begin{theorem}\label{thm:main}
It is decidable whether an Orbit instance admits a semialgebraic invariant.
Furthermore, there exists an algorithm which constructs such an invariant when it exists,
and the invariant produced has polynomial-size description.
\end{theorem}

The remainder of the paper is devoted to proving Theorem~\ref{thm:main}. To this end, 
let $\ell = (A,x,y)$ be a non-reachability Orbit instance in dimension
$d$.\footnote{Kannan and Lipton showed the decidability of
  reachability for Orbit instances over rational numbers; their proof
  carries over to instances with algebraic entries, however without
  the polynomial-time complexity.}

As a first step, recall that every matrix $A$ can be written in the
form $A=Q^{-1} J Q$, where $Q$ is invertible and $J$ is in Jordan
normal form.  The following lemma transfers semialgebraic invariants
through the change-of-basis matrix $Q$.

\begin{lemma}\label{lem:basis}
Let $\ell = (A,x,y)$ be an Orbit instance, and $Q$ an invertible matrix in $\A^{d \times d}$.

Construct the Orbit instance $\ell_Q = (Q A Q^{-1},Q x, Q y)$.
Then $\P$ is a semialgebraic invariant for $\ell_Q$
if, and only if, $Q^{-1} \P$ is a semialgebraic invariant for $\ell$.
\end{lemma}

\begin{proof}
First of all, $Q^{-1} \P$ is semialgebraic if, and only if, $\P$ is semialgebraic.
We have:
\begin{itemize}
\item $Q A Q^{-1} \P \subseteq \P$ if, and only if, $A Q^{-1} \P \subseteq Q^{-1} \P$,
  	\item $Q x \in \P$ if, and only if, $x \in Q^{-1} \P$,
	\item $Q y \notin \P$, if, and only if, $y \notin Q^{-1} \P$.
\end{itemize} 
This concludes the proof.
\end{proof}

Thanks to Lemma~\ref{lem:basis}, we can reduce the problem of the
existence of semialgebraic invariants for Orbit instances to cases
in which the matrix is in Jordan normal form, \emph{i.e.}, is a diagonal
block matrix, where the blocks (called Jordan blocks) are of the form:
\[
\begin{bmatrix}
\lambda & 1            & \;     & \;  \\
\;        & \lambda    & \ddots & \;  \\
\;        & \;         & \ddots & 1   \\
\;        & \;         & \;     & \lambda       
\end{bmatrix}
\]
Note that this transformation can be achieved in polynomial time~\cite{Cai00,CLZ00}.

Formally, a Jordan block is a matrix $\lambda I + N$ with $\lambda \in
\A$, $I$ the identity matrix and $N$ the matrix with $1$'s on the
upper diagonal, and $0$'s everywhere else.  The number $\lambda$ is an
eigenvalue of $A$.  A Jordan block of dimension one is called
diagonal, and $A$ is diagonalisable if, and only if, all Jordan blocks
are diagonal.

The $d$ dimensions of the matrix $A$ are indexed by pairs $(J,k)$,
where $J$ ranges over the Jordan blocks and $k \in
\set{1,\ldots,\delta}$ where $\delta$ is the dimension of the Jordan
block $J$.  For instance, if the matrix $A$ has two Jordan blocks,
$J_1$ of dimension $1$ and $J_2$ of dimension $2$, then the three
dimensions of $A$ are $(J_1,1)$ (corresponding to the Jordan block
$J_1$) and $(J_2,1),(J_2,2)$ (corresponding to the Jordan block
$J_2$).

For a vector $v$ and a subset $S$ of $\set{1,\ldots,d}$, we denote $v_S$ the projection vector of $v$ on the dimensions in $S$,
and extend this notation to matrices.
As a special case, $v_{J,>k}$ denotes the vector restricted to the coordinates of the Jordan block $J$ whose index is greater than $k$.
We denote $\overline{S}$ the complement of $S$ in $\set{1,\ldots,d}$.

\medskip
There are a few degenerate cases which we handle now.
We say that an Orbit instance $\ell = (A,x,y)$ in Jordan normal form is non-trivial if:
\begin{itemize}
	\item There is no Jordan block associated with the value $0$, or equivalently $A$ is invertible,
	\item For each Jordan block $J$, both $x_J$ and $y_J$ are not the zero vector,
	\item For each non-diagonal Jordan block $J$, the vector $x_J$ has at least a non-zero coordinate other than the first one,
	\emph{i.e.}, $x_{J,>1}$ is not the zero vector.
\end{itemize}

\begin{lemma}\label{lem:reduc}
The existence of semialgebraic invariants for Orbit instances reduces in poylnomial time 
to the same problem for non-trivial Orbit instances in Jordan normal form.
\end{lemma}

\begin{proof}
Let $\ell = (A,x,y)$ be an Orbit instance in Jordan normal form.

\begin{itemize}
	\item If $A$ is not invertible, we distinguish two cases.
		\begin{itemize}
			\item If for some Jordan block $J$ associated with the eigenvalue $0$, we have that $y$ is not the zero vector,
			\emph{i.e.}, $y_J \neq 0$,
			then consider $\P = \set{x, Ax, \ldots, A^{d-1}x} \cup \set{z \in \C^d \mid z_J = 0}$ is a semialgebraic invariant.
			Indeed, the Jordan block $J$ is nilpotent, so for any vector $u$ and $n \ge d$, we have that $J^n u = 0$,
			so in particular $(A^n x)_J = 0$.
			Moreover, since by assumption $y$ is not reachable, it is not one of $A^n x$ for $n < d$, 
			and $y_J \neq 0$, so $y \notin \P$.			
			
			\item Otherwise, denote $J$ the dimensions corresponding to Jordan blocks associated with the eigenvalue $0$,
			we have that $y_J = 0$.
			Consider the Orbit instance $\ell_J = (A_{\overline{J}},(A^d x)_{\overline{J}},y_{\overline{J}})$.
			We claim that $\ell$ admits a semialgebraic invariant if, and only if, $\ell_J$ does.

			Let $\P$ be a semialgebraic invariant for $\ell$.
			Construct $\P_J$ the set of vectors $z$ in $\C^{\overline{J}}$ such that $z$ augmented with $0$'s in the $J$ dimensions yields a vector in $\P$,
			we argue that $\P_J$ is a semialgebraic invariant for $\ell_J$.
			Indeed, $(A^d x)_{\overline{J}} \in \P_J$ since $A^d x \in \P$
			and $(A^dx)_J=0$, because the Jordan block $J$ is nilpotent.
			The stability of $\P_J$ under $A_{\overline{J}}$ is clear, and $y_{\overline{J}} \notin \P_J$
			because $y_J = 0$, so $y_{\overline{J}} \in \P_J$ would imply $y \in \P$.

			Conversely, let $\P_J$ be a semialgebraic invariant for $\ell_J$, 
			extend it to $\P \subseteq \C^d$ by allowing any complex numbers in the $J$ dimensions,
			then $\set{x,Ax,\ldots,A^{d-1}x} \cup \P$ is a semialgebraic invariant for $\ell$.

			We reduced the existence of semialgebraic invariants from $\ell$ to $\ell_J$, with the additional property that the matrix is invertible.
		\end{itemize}

	\item Suppose $A$ contains a Jordan block $J$ such that either $x_J = 0$ or $y_J = 0$.  We distinguish three cases.
		\begin{itemize}
			\item If for some Jordan block $J$ we have $x_J = 0$ and $y_J \neq 0$, 
			then $\P = \set{z \in \C^d \mid z_J = 0}$ is a semialgebraic invariant for $\ell$.
	
			\item If for some Jordan block $J$ we have $x_J \neq 0$ and $y_J = 0$,
			let $k$ such that $x_{J,k} \neq 0$ and $x_{J,>k} = 0$,
			then $\P = \set{z \in \C^d \mid z_{J,k} \neq 0 \text{ and } z_{J,>k} = 0}$ is a semialgebraic invariant for $\ell$.
	
			\item Otherwise, denote $J$ the dimensions corresponding to Jordan blocks for which $x_J = y_J = 0$.
			Consider the Orbit instance $\ell_J = (A_{\overline{J}},x_{\overline{J}},y_{\overline{J}})$,
			we claim that $\ell$ admits a semialgebraic invariant if, and only if, $\ell_J$ does.
			
			Let $\P$ be a semialgebraic invariant for $\ell$.
			Construct $\P_J$ the set of vectors $z$ in $\C^{\overline{J}}$ such that $z$ augmented with $0$ in the $J$ dimensions yields
			a vector in $\P$, then $\P_J$ is a semialgebraic invariant for $\ell_J$.
			
			Conversely, let $\P_J$ be a semialgebraic invariant for $\ell_J$,
			extend it to $\P \subseteq \C^d$ by allowing only $0$ in the $J$ dimensions,
			then $\P$ is a semialgebraic invariant for $\ell$.

			We reduced the existence of semialgebraic invariants from $\ell$ to $\ell_J$, with the additional property that
			for each Jordan block $J$, both $x_J$ and $y_J$ are not the zero vector.
		\end{itemize}

	\item If $A$ contains a non-diagonal Jordan block $J$ such that the vector $x_J$ is zero except on the first coordinate $(J,1)$,
		we distinguish two cases.
		\begin{itemize}
			\item If for some non-diagonal Jordan block $J$ we have that $x_{J,>1} = 0$ and $y_{J,>1} \neq 0$,
			 then $\P = \set{z \in \C^d \mid z_{J,>1} = 0}$ is a semialgebraic invariant for $\ell$.
			\item Otherwise, denote $J$ the dimensions corresponding to non-diagonal Jordan blocks for which $x_{J,>1} = y_{J,>1} = 0$.
			Let $S = \overline{J} \cup \bigcup_J (J,1)$, \emph{i.e.}, the dimensions outside $J$ plus the first dimensions of each block in $J$.
			Consider the Orbit instance $\ell_S = (A_S,x_S,y_S)$,
			we claim that $\ell$ admits a semialgebraic invariant if, and only if, $\ell_S$ does.

			Let $\P$ be a semialgebraic invariant for $\ell$.
			Construct $\P_S$ the set of vectors $z$ in $\C^S$ such that $z$ augmented with $0$ in the $\overline{S}$ dimensions yields
			a vector in $\P$, then $\P_S$ is a semialgebraic invariant for $\ell_S$.
			
			Conversely, let $\P_S$ be a semialgebraic invariant for $\ell_S$,
			extend it to $\P \subseteq \C^d$ by allowing only $0$ in the $\overline{S}$ dimensions,
			then $\P$ is a semialgebraic invariant for $\ell$.

			We reduced the existence of semialgebraic invariants from $\ell$ to $\ell_S$, with the additional property that
			for each non-diagonal Jordan block $J$, $x_{J,>1}$ is not the zero vector.			
		\end{itemize}		
\end{itemize}
This concludes the proof.
\end{proof}

\subsection{Some eigenvalue has modulus different from one} 
\label{sec:neqone}
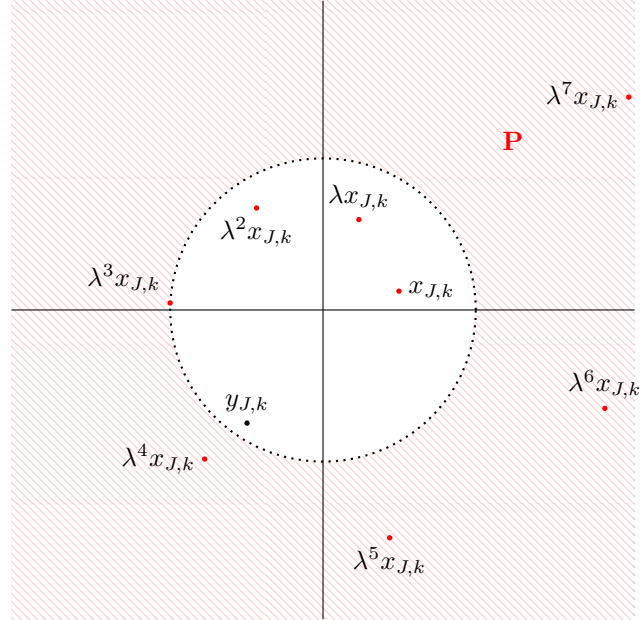
\begin{figure}[!ht]
	\centering
	\begin{tikzpicture}[scale=5/10]
	\def\ptsize{1.0pt}
	\coordinate[label=above left:O] (O) at (0,0);
	
	\coordinate (north) at (0,8.2);
	\coordinate (south) at (0,-8.2);
	\coordinate (east) at (8.2,0);
	\coordinate (west) at (-8.2,0);

	\draw[dotted, thick, fill=white]
		(0,0) circle[x radius = 4.02 cm, y radius = 4.02 cm];

	\coordinate (ne) at (8.2,8.2);
	\coordinate (nw) at (-8.2,8.2);
	\coordinate (sw) at (-8.2,-8.2);
	\coordinate (se) at (8.2,-8.2);

	\coordinate[label=right:$x_{J,k}$] (x0) at (2,0.5);
	\coordinate[label=above:$\lambda x_{J,k}$] (x1) at (0.945,2.397);
	\coordinate[label=below:$\lambda^2 x_{J,k}$] (x2) at (-1.749,2.704);
	\coordinate[label=above left:$\lambda^3 x_{J,k}$] (x3) at (-4.022,0.187);
	\coordinate[label=left:$\lambda^4 x_{J,k}$] (x4) at (-3.115,-3.953);
	\coordinate[label=below:$\lambda^5 x_{J,k}$] (x5) at (1.754,-6.041);
	\coordinate[label=above:$\lambda^6 x_{J,k}$] (x6) at (7.418,-2.609);
	\coordinate[label=left:$\lambda^7 x_{J,k}$] (x7) at (8.047,5.645);
	
	\coordinate[label=above:\textcolor{red}{$\mathbf P$}] (P) at (5,4);

	\coordinate[label=above:$y_{J,k}$ ] (y) at (-2,-3);
	\draw (south) -- (north);
	\draw (west) -- (east);
	
	\fill (y) circle (2.0pt);
	
	\foreach \p in {x0, x1, x2, x3, x4, x5, x6, x7}
		\fill[red] (\p) circle (2.0pt);
	
		\begin{scope}[on background layer]
	\fill[pattern=north west lines, pattern color=red!20] (ne.center)--(nw.center)--(sw.center)--(se.center);
	\end{scope}
\end{tikzpicture}
	\caption{Case $|\lambda| > 1$. This figure represents the complex plane, which is the projection on the coordinate $(J,k)$.}
	\label{fig:lambdagr1semial}
\end{figure}

\begin{lemma}
Let $\ell = (A,x,y)$ be a non-trivial Orbit instance in Jordan normal form.
Assume that $\ell$ is a non-reachability instance.
If the matrix $A$ has an eigenvalue whose modulus is not equal to $1$,
then there exists a semialgebraic invariant for $\ell$.
\end{lemma}

\begin{proof}
We distinguish two cases according to whether there exists an
eigenvalue of modulus strictly more than $1$ or en eigenvalue of
modulus strictly less than $1$.
\begin{itemize}
	\item Suppose that $A$ contains a Jordan block $J$ associated with an eigenvalue $\lambda$ with $|\lambda| > 1$. 

	In this case, some coordinate of $(A^n x)_{n \in \N}$ diverges
        to infinity, so eventually gets larger in modulus than the corresponding
        coordinate in $y$.  This allows us to construct a
        semialgebraic invariant for $\ell$ by taking the first points
        and then all points having a large coordinate in the diverging
        dimension.  This case is illustrated in Figure
        \ref{fig:lambdagr1semial}.
	
	By assumption $x_J$ is non-zero, let $k$ such that $x_{J,k} \neq 0$ and $x_{J,>k} = 0$
	(Note that if the Jordan block $J$ is diagonal, $k = 1$).
	For all $n \in \N$, we have $(A^n x)_{J,k} = \lambda^n x_{J,k}$, so $|(A^n x)_{J,k}|$ diverges to infinity.
	It follows that there exists $n_0 \in \N$ such that $|(A^{n_0} x)|_{J,k} > |y_{J,k}|$.
	Let
\[
\P = \set{x, A x,\ldots,A^{n_0-1} x} \cup \set{z \in \C^d \mid |z_{J,k}| \geq |(A^{n_0} x)_{J,k}| \text{ and } z_{J,>k} = 0}.
\]

	We argue that $\P$ is a semialgebraic invariant for $\ell$.
	The non-trivial point is that $\P$ is stable under $A$.
	Note that $(A^{n_0} x)_{J,>k} = 0$, so $A^{n_0} x \in \P$.
	Let $z \in \C^d$ such that $|z_{J,k}| \geq |(A^{n_0} x)_{J,k}|$ and $z_{J,>k} = 0$.
	Then $(A z)_{J,k} = \lambda z_{J,k}$ and $(A z)_{J,>k} = 0$, so $A z \in \P$.

	\item If $A$ contains a Jordan block $J$ associated with an eigenvalue $\lambda$ with $|\lambda| < 1$. 

	The situation is similar to the former, except that the convergence is towards the origin.
	The construction of the semialgebraic invariant is much more subtle though, for the following reason:
	for $k$ such that $x_{J,k} \neq 0$ and $x_{J,>k} = 0$, we may have that $y_{J,k} = 0$,
	implying that $((A^n x)_{J,k})_{n \in \N}$ does not become smaller than $y_{J,k}$. 
	Working on another dimension implies to give up the following diagonal behaviour: $(A^n x)_{J,k} = \lambda^n x_{J,k}$,
	making it hard to find a stable set under $A$.
	To overcome this problem, the invariant we define depends upon all the coordinates of the Jordan block $J$.

	Denote $d(J)$ the dimension of the Jordan block $J$.
	We have that $((A^n x)_J)_{n \in \N}$ converges to $0$. 
	It follows that there exists $n_0 \in \N$ such that for each dimension $k$ of the Jordan block $J$,
	\emph{i.e.}, for $k \in \set{1,\ldots,d(J)}$, we have $|(A^{n_0} x)_{J,k}| \le (1 - |\lambda|)^k \cdot ||y_J||_\infty$.

	Let
\[
\P = \set{x, A x,\ldots,A^{n_0-1} x} \cup \set{z \in \C^d \mid \forall k \in \set{1,\ldots,d(J)}, |z_{J,k}| \leq (1 - |\lambda|)^k \cdot ||y_J||_\infty}.
\]

	We argue that $\P$ is a semialgebraic invariant for $\ell$.
	Note that $y \notin \P$ since for $k$ such that $||y_J||_\infty = |y_{J,k}|$, this would imply $||y_J||_\infty \leq (1 - |\lambda|)^k \cdot ||y_J||_\infty$,
	which cannot be since $k \ge 1$, $y_J \neq 0$ and $|\lambda| < 1$.
	We examine the stability of $\P$ under $A$.
	Let $z \in \C^d$ such that for each dimension $k \in \set{1,\ldots,d(J)}$, we have $|z_{J,k}| \leq (1 - |\lambda|)^k \cdot ||y_J||_\infty$.
	Let $k < d(J)$, then
\[
\begin{array}{lll}
|(Az)_{J,k}| = |\lambda z_{J,k} + z_{J,k+1}| 
& \leq & |\lambda| |z_{J,k}| + |z_{J,k+1}| \\
& \leq & |\lambda| (1 - |\lambda|)^k \cdot ||y_J||_\infty + (1 - |\lambda|)^{k+1} \cdot ||y_J||_\infty \\
& =    & (|\lambda| + (1 - |\lambda|))(1 - |\lambda|)^k \cdot ||y_J||_\infty \\
& =    & (1 - |\lambda|)^k \cdot ||y_J||_\infty.
\end{array}
\]
	The case $k = d(J)$ is similar but easier.	
\end{itemize}
This concludes the proof.
\end{proof}

\subsection{All eigenvalues have modulus one and the matrix is not diagonalisable} 
\label{sec:nondiag}
\begin{lemma}
Let $\ell = (A,x,y)$ be a non-trivial Orbit instance in Jordan normal
form and assume that $\ell$ is a non-reachability instance.  If all
the eigenvalues of the matrix $A$ have modulus $1$ and $A$ is not
diagonalisable, then there exists a semialgebraic invariant for
$\ell$.
\end{lemma}

We illustrate the construction of the semialgebraic invariant in an
example following the proof.  (See also Example~\ref{ex:three} from
the Introduction.)

\begin{proof}
By assumption, there exists a non-diagonal Jordan block $J$.
Since $\ell$ is non-trivial, $x$ has a non-zero coordinate in $J$ which is not the first one. 
Let $k$ such that $x_{J,k} \neq 0$ and $x_{J,>k} = 0$, we have $k \ge 2$ and
\[
(A^n x)_{J,k-1} = \lambda^n x_{J,k-1} + n \lambda^{n-1} x_{J,k},
\]
so $(|(A^n x)_{J,k-1}|)_{n \in \N}$ diverges to infinity since $|\lambda| = 1$.
It follows that there exists $n_0 \in \N$ such that $|(A^{n_0} x)_{J,k-1}| > |y_{J,k-1}|$.
Without loss of generality we assume $n_0 \geq -\frac{\Dot{\lambda x_{J,k-1},x_{J,k}}}{|x_{J,k}|^2}$. 
The notation $\Dot{u,v}$ designates the scalar product of the complex numbers $u$ and $v$ viewed
as vectors in $\mathbb{R}^2$, defined by $\Re{u \overline{v}}$.
This quantity will appear later; note that
%it is constant, meaning that
it only depends on $x$ and $A$.

Let 
\[
\P = \set{x, Ax, \ldots, A^{n_0-1} x} \cup 
\set{z \in \C^d \left|
\begin{array}{l}
|z_{J,k-1}| \geq |(A^{n_0} x)_{J,k-1}| \text{, and } \\
\Dot{\lambda z_{J,k-1},z_{J,k}} \geq 0 \text{, and } z_{J,>k} = 0
\end{array}\right.
}.
\]

We argue that $\P$ is a semialgebraic invariant for $\ell$.
It is a semialgebraic set: the condition $\Dot{\lambda z_{J,k-1},z_{J,k}} \geq 0$ is of the form $P(z) \ge 0$ for a polynomial $P$ with algebraic coefficients,
where $z$ is seen as a vector in $\R^{2d}$.
The part to be looked at closely is the stability of $\P$ under $A$.

First, $A^{n_0} x \in \P$. 
Indeed, using $|\lambda|=1$ and the assumption on $n_0$:
\[
\begin{array}{lll}
\Dot{\lambda(A^{n_0} x)_{J,k-1} , (A^{n_0} x)_{J,k}}
& = & \Dot{\lambda \cdot \left(\lambda^{n_0} x_{J,k-1} + n_0 \lambda^{n_0-1} x_{J,k} \right) , \lambda^{n_0} x_{J,k}} \\ 
& = & |\lambda^{n_0}|^2 \Dot{\lambda x_{J,k-1} , x_{J,k}} + n_0 |\lambda^{n_0} x_{J,k}|^2\\
& = & \Dot{\lambda x_{J,k-1} , x_{J,k}} + n_0 |x_{J,k}|^2\\
& \geq & 0.
\end{array}
\]

Now, let $z \in \C^d$ such that $|z_{J,k-1}| \geq |(A^{n_0} x)_{J,k-1}|$, $\Dot{\lambda z_{J,k-1},z_{J,k}} \geq 0$ and $z_{J,>k} = 0$.
We have $(A z)_{J,k-1} = \lambda z_{J,k-1} + z_{J,k}$, $(A z)_{J,k} = \lambda z_{J,k}$ and $(A z)_{J,>k} = 0$.
It follows that:
\[
\begin{array}{lll}
\left|(A z)_{J,k-1}\right|^2 & = & \left|\lambda z_{J,k-1} + z_{J,k} \right|^2 \\
& = & |z_{J,k-1}|^2 + 2 \langle \lambda z_{J,k-1}, z_{J,k} \rangle + |z_{J,k}|^2 \\
& \geq & |z_{J,k-1}|^2 \\
& \geq & \left|(A^{n_0} x)_{J,k-1} \right|^2,
\end{array}
\]
and:
\[
\begin{array}{lll}
\langle \lambda(A z)_{J,k-1}, (A z)_{J,k} \rangle & = & \langle \lambda(\lambda z_{J,k-1} + z_{J,k}), \lambda z_{J,k} \rangle \\
& = & |\lambda|^2 \langle \lambda z_{J,k-1} +  z_{J,k}, z_{J,k} \rangle \\
& = & \langle \lambda z_{J,k-1}, z_{J,k} \rangle + |z_{J,k}|^2 \\
& \geq & 0.
\end{array}
\]
Hence $A z \in \P$, and $\P$ is a semialgebraic invariant for $\ell$.
\end{proof}

\begin{example}
Consider the following matrix:
\[
A = \begin{bmatrix}
e^{i \theta} & 1 \\
0 & e^{i \theta} \\
\end{bmatrix},
\]
where $\theta \in \R$ is an angle such that $\frac{\theta}{\pi} \notin \Q$.
We start from the vector $x = \left[1,\ 1\right]^T$.
We have
\[
A^n x = \left[e^{in\theta} + n e^{i(n-1)\theta},\ e^{in\theta} \right],
\]
so the projection on the second coordinate is a dense subset of the
unit circle, and the projection on the first coordinate describes a
growing spiral (similar to that shown in
Figure~\ref{fig:lambdagr1semial}).  A tentative invariant for
excluding some vector $y$ is the complement of a circle on the first
coordinate, large enough not to include $y$. However, this set is not
a priori invariant.  Geometrically, the action of $A$ on a vector
$\left[z_1,\ z_2\right]$ is to rotate both $z_1$ and $z_2$ by an angle
of $\theta$, and to push the first coordinate in the direction of
$z_2$:
\[
A \left[z_1,\ z_2\right] = \left[e^{i\theta} z_1 + z_2,\ e^{i\theta} z_2\right].
\]
A natural way to restrict the above set to make it invariant is to ensure that $z_2$ pushes away from the origin,
\emph{i.e.}, that the norm of $(A z)_1$ increases. 
This is achieved by requiring that $\Dot{e^{i\theta} z_1, z_2} \ge 0$.
\end{example}

\subsection{All eigenvalues have modulus one and the matrix is diagonalisable} 
\label{sec:eqone}
This case is the most involved and is the only one in which it might
hold that $y$ not be reachable and yet no semialgebraic invariant
exists.  (Recall Example~\ref{ex:one} from the Introduction.)  Using results
from Diophantine approximation and algebraic number theory, we show
that the topological closure of the orbit $\overline{\set{A^nx:
x\in\mathbb{N}}}$ is (effectively) semialgebraic.  Furthermore, using
topological properties of semialgebraic sets we show that any
semialgebraic invariant must contain the closure of the orbit.  It
follows that there exists a semialgebraic invariant just in case
$y\not\in \overline{\set{A^nx: x\in\mathbb{N}}}$.

We start with the following topological fact about semialgebraic sets.

\begin{lemma}\label{fact}
Let $E,F \subseteq \mathbb{R}^n$ be two sets such that $\overline{E}
= \overline{F}$ and $F$ is semialgebraic.  Then $E \cap
F \neq \emptyset$.
\end{lemma}

\begin{proof}
The proof uses the notion of the dimension of a semialgebraic set.
The formal definition of dimension uses the cell-decomposition
theorem (see, e.g.,~\cite[Chapter 4]{vDries98}).  However
to establish the lemma it suffices to note the following two
properties of the dimension.  First, for any semi-algebraic set
$X\subseteq \mathbb{R}^n$ set we have
$\dim(X)=\dim(\overline{X})$~\cite[Chapter 4, Theorem 1.8]{vDries98}.
Secondly, if $X \subseteq Y$ are semi-algebraic subsets of
$\mathbb{R}^n$ that have the same dimension, then $X$ has non-empty
interior in $Y$~\cite[Chapter 4, Corollary 1.9]{vDries98}.

In the situation at hand, since $\dim(F)=\dim(\overline{F})$ it
follows that $F$ has non-empty interior (with respect to the subspace
topology) in $\overline{F}=\overline{E}$.  But then $E$ is dense in
$\overline{E}$ while $F$ has non-empty interior in $\overline{E}$, and
thus $E$ and $F$ meet.
\end{proof}

\begin{lemma}
  Let $\ell = (A,x,y)$ be an Orbit instance, where
  $A=\mathrm{diag}(\lambda_1,\ldots,\lambda_d)$ is a diagonal
  $d\times d$ matrix with entries
  $\lambda_1,\ldots,\lambda_d\in \mathbb{C}$ all having modulus one.
  Write $\mathcal{O}=\{ A^n x : n \in \mathbb{N} \}$ for the orbit of
  $x$ under $A$.  Then
\begin{itemize}
\item The topological closure of $\mathcal{O}$ in $\C^d$ is a
  semi-algebraic set that is computable from $\ell$ in polynomial
  space.
\item Any semi-algebraic invariant for $\ell$ contains
  $\overline{\mathcal{O}}$.
\end{itemize}
\label{lem:semi-alg}
\end{lemma}

\begin{proof}
We start by proving the first item.

Write $\mathbb{T}$ for the unit circle in $\C$.
	Let
        \[ L_A = \set{v \in \Z^d \mid \lambda_1^{v_1} \cdots
            \lambda_d^{v_d} = 1}\] be the set of all multiplicative
        relations holding among $\lambda_1,\ldots,\lambda_d$.  Notice
that $L_A$ is an additive subgroup of $\Z^d$.   Consider the 
set of diagonal $d\times d$ matrices
\[ {T}_A = \set{ \mathrm{diag}(\mu_1,\ldots,\mu_d) \mid \mu \in
    \mathbb{T}^d \mbox{ and } \forall v \in L_A \, ( \mu_1^{v_1} \cdots
    \mu_d^{v_d} = 1)} \] whose diagonal entries satisfy the
multiplicative relations in $L_A$.  Notice that $T_A$ forms a group under
matrix multiplication that is also a closed subset of $\mathbb{C}^{d\times d}$.

Using Kronecker's Theorem on inhomogeneous simultaneous Diophantine
approximation~\cite{C65}, it is shown in~\cite[Proposition 3.5]{OW14}
that $\{ A^n : n \in \mathbb{N}\}$ is a dense subset of ${T}_A$.  This
immediately gives
\begin{gather}
\overline{\mathcal{O}} = \overline{\{A^nx:n\in\mathbb{N}\}} 
= \{Mx: M \in {T}_A\} \, .
\label{eq:closure}
\end{gather}

We now show that $\overline{\O}$ is semi-algebraic.  Observe that
$L_A$ is finitely generated, being a subgroup of a finitely generated
group.  Moreover, if $B \subseteq L_A$ is a basis of $L_A$ then we can
write
\[ {T}_A = \set{ \mathrm{diag}(\mu_1,\ldots,\mu_d) \mid \mu \in
    \mathbb{T}^d \mbox{ and } \forall v \in B\, (\mu_1^{v_1} \cdots
    \mu_d^{v_d} = 1)} \, . \] It follows that ${T}_A$ is a
semi-algebraic subset of $\mathbb{C}^{d\times d}$ and thus
from  (\ref{eq:closure}) that $\overline{\O}$ is a
semi-algebraic set.  

From an upper bound on the length of $B$ due to Masser~\cite{M88}, it
can be shown that one can compute a basis for $L_A$ in polynomial
space in the description of $A$ (see~\cite[Corollary 3.3]{OW14}) and
thereby compute a description of ${T}_A$ as a semi-algebraic set, also
in polynomial space in the description of $A$.

Now we move to the second item in the statement of the lemma.
Let $\P$ be a semi-algebraic invariant for $\ell$.  Our goal is
  to show that $\overline{\O} \subseteq \P$.  To show this we can,
  without loss of generality, replace $\P$ by $\P \cap \overline{\O}$, since the
  latter is also a semi-algebraic invariant.  Moreover, since any
  invariant necessarily contains the orbit $\O$, we may suppose
  that $\O \subseteq \P \subseteq \overline{\O}$, and hence
  $\overline{\P}=\overline{\O}$.

  We now prove that $\overline{\O} \subseteq \P$, that is, we pick an
  arbitrary element $z \in \overline{\O}$ and show that $z\in \P$.  To
  this end, consider the orbit of $z$ under the matrix $A^{-1}$.  Now
  $A^{-1}=\mathrm{diag}(\lambda_1^{-1},\ldots,\lambda_d^{-d})$ and we
  may define groups $L_{A^{-1}}$ and $T_{A^{-1}}$ analogously with
  $L_A$ and $T_A$.  In fact it is clear that $L_{A}$ and
  $L_{A^{-1}}$ coincide (i.e., $\lambda_1,\ldots,\lambda_d$ satisfy
  exactly the same multiplicative relations as
  $\lambda^{-1}_1,\ldots,\lambda^{-1}_d$), and hence
  also ${T}_A = {T}_{A^{-1}}$.

Now we claim that the following chain of equalities holds:
\begin{eqnarray}
\overline{  \{ A^{-n} z : n \in \mathbb{N} \} }
&=&
\{ Mz : M \in {T}_{A^{-1}} \} \label{eq:one} \\
&=&
\{ Mz : M \in {T}_{A} \} \label{eq:two} \\
&=& 
\{ Mx : M \in {T}_{A} \} \label{eq:three} \\
&= & 
\overline{\O} \; = \; \overline{\mathcal{P}} \,  . \notag
\end{eqnarray}
Indeed, Equation (\ref{eq:one}) is an instance of (\ref{eq:closure}),
but with $A^{-1}$ and $z$ in place of $A$ and $x$.  Equation
(\ref{eq:two}) follows from the fact that ${T}_A = {T}_{A^{-1}}$.  To
see Equation (\ref{eq:three}), observe from (\ref{eq:closure}) that
$z$ has the form $M_0x$ for some $M_0 \in {T}_A$.  But
$\{MM_0 x : M \in {T}_A \} = \{Mx : M \in {T}_A \}$ since ${T}_A$,
being a group, contains $M_0^{-1}$.

Now we have established that 
\[ \overline{ \{ A^{-n} z : n \in \mathbb{N} \} } =
  \overline{\P} \, .\] Then by Lemma~\ref{fact} we have that
$A^{-n} z$ lies in $\P$ for some $n \in \mathbb{N}$.  But since $\P$
is invariant under $A$ we have $z \in \P$.
\end{proof}

\begin{corollary}
  Let the Orbit instance $\ell$ be as described in Lemma~\ref{lem:semi-alg}.
  Then $\ell$ admits a semi-algebraic invariant if and only if
  $y \notin \overline{\O}$.
\end{corollary}
\begin{proof}
  If $y \notin \overline{\O}$, then $\overline{\O}$ is a
  semi-algebraic invariant for $\ell$ by the first item in
  Lemma~\ref{lem:semi-alg}.  Conversely, if there exists a
  semi-algebraic invariant $\P$ for $\ell$, then
  $\overline{\O}\subseteq \P$ by the second item in
  Lemma~\ref{lem:semi-alg}, implying that $y \notin \overline{\O}$.
\end{proof}

%\begin{example}
%Consider the diagonal matrix $A = \mathrm{diag}(e^{i \theta},e^{i \theta'})$ of dimension $2$.
%We start from the vector $x = \left[1,\ 1\right]^T$, and ask whether some vector $y$ is reachable.
%Note that this is equivalent to asking whether there exists $n \in \N$ such that $e^{i n\theta} + e^{i n\theta'} = y$.
%
%On both coordinates, the matrix $A$ acts as a rotation, of angle $\theta$ in the first coordinate and $\theta'$ in the second.
%If $\frac{\theta}{\pi} \notin \Q$ and $\frac{\theta'}{\pi} \notin \Q$, both projections of the orbit are dense in the unit circle.
%Since any semi-algebraic invariant contains the orbit, it must also contain the unit circle, hence is of no help to determine whether $y$ is reached or not
%when $y$ lies on the unit circle.
%Indeed, complicated dependencies between $\theta$ and $\theta'$ may prevent from reaching $y$, and this cannot be captured by semi-algebraic invariants.
%\end{example}

\subsection{Proof of Theorem~\ref{thm:main}}
We now draw together the results of the previous sections to prove our
main result, Theorem~\ref{thm:main}, giving an effective
characterisation of the existence of semialgebraic invariants and a
procedure to compute such an invariant when it exists.

Let $\ell = (A,x,y)$ be a non-reachability Orbit instance.  First we put
$A$ in Jordan normal form and simplify $\ell$ to obtain a non-trivial
Orbit instance.  We then divide into three cases.
\begin{itemize}
        \item If some eigenvalue of $A$ has modulus different from
          $1$ then there is a semialgebraic invariant (see Section~\ref{sec:neqone}).
        \item If all eigenvalues have modulus $1$ and the matrix is
          not diagonalisable then there is a semialgebraic invariant (see Section~\ref{sec:nondiag}).
        \item If all eigenvalues have modulus $1$ and the matrix is
          diagonalisable, then there exists a semialgebraic invariant
          if and only if the topological closure of the orbit
          $\overline{\set{A^nx:n\in\mathbb{N}}}$ is such an invariant,
          which holds if and only if the closure does not contain $y$
          (see Section~\ref{sec:eqone}). Note therefore that
          non-reachability Orbit instances for which there do not
          exist semialgebraic invariants are extremely sparse.
\end{itemize}
Thus we obtain an effective characterisation of the class of Orbit instances for
which there exists a semialgebraic invariant.  Moreover in those cases
in which there exists an invariant we have shown how to compute such
an invariant in polynomial space.

\subsection{On topologically closed invariants}

We now shortly discuss the point of view of topologically closed semialgebraic invariants. It is worth noting that the only non-closed invariant we synthesise is in case 2.2 of lemma \ref{lem:reduc}. This allows to state a more concise result, although a bit weaker. We state it as a corollary.

\begin{corollary}
A topologically closed semialgebraic invariant for the loop $\ell = (A,x,y)$ exists if and only if $y \notin \overline{\set{A^nx:x\in\mathbb{N}}}$ the topological closure of the orbit.
\end{corollary}

A full proof will appear in the journal version of this paper.

\section{Conclusions}
%We wrap up the proof of Theorem~\ref{thm:main}.
%Let $\ell = (x,A,y)$ be a linear loop, we put $A$ in Jordan normal form and simplify $\ell$ to obtain a non-trivial linear loop, and then:
%\begin{itemize}
%	\item If some eigenvalue of $A$ has modulus different from $1$, then there exists a semi-algebraic invariant,
%	\item If all eigenvalues have modulus $1$ and the matrix is not diagonalisable, then there exists a semi-algebraic invariant,
%	\item If all eigenvalues have modulus $1$ and the matrix is diagonalisable, then there exists a semi-algebraic invariant if, and only if,
%	the topological closure of the orbit is such a set, which is equivalent to asking whether it contains $y$.
%\end{itemize}

%In most cases, when $\ell$ does not terminate this is witnessed by a semi-algebraic invariant. Furthermore, in such a case, there exists a semi-algebraic invariant
%of polynomial space description.

\medskip
This paper is a first step towards the study of invariants for
discrete linear dynamical systems. At present, the question of the existence and
of the algorithmic synthesis of suitable invariants for
higher-dimensional versions of the Orbit Problem (i.e., when the
`target' $y$ to be avoided consists of either a vector space, a
polytope, or some other higher-dimensional object) is completely
open. Given, as pointed out earlier, that reachability questions with
high-dimensional targets appear themselves to be very difficult, one
does not expect the corresponding invariant synthesis problems to be
easy, yet this approach might prove a tractable alternative well worth
exploring.

Our main result is a polynomial-space procedure for deciding existence
and computing semialgebraic invariants in instances of the Orbit
Problem.  The only obstacle to obtaining a polynomial-time bound is
the problem of computing a basis of the group of all multiplicative
relations among a given collection of algebraic numbers
$\alpha_1,\ldots,\alpha_d$, which is not known to be solvable in
polynomial time.  Less ambitiously one can ask for a polynomial-time
procedure to verify a putative relation
$\alpha_1^{n_1}\ldots \alpha_d^{n_d}\stackrel{?}{=}1$.  Assuming that
$\alpha_1,\ldots,\alpha_d$ are represented as elements of an
explicitly given finite-dimensional algebra $K$ over $\mathbb{Q}$,
Ge~\cite{Ge93:FOCS} gave a polynomial-time algorithm for verifying
multiplicative relations.  In our setting, however, where
$\alpha_1,\ldots,\alpha_d$ are roots of the characteristic polynomial
of matrix $A$, the dimension of $K$ may be exponential in $d$.

\bibliography{biblio}

\end{document}